\documentclass[peerreview]{IEEEtran}
\usepackage{algorithm,algpseudocode,amsmath,amsfonts,amssymb,amsthm,array,color,cite,comment,caption,epsfig, subfigure}
\newtheorem{lemma}{Lemma}
\newtheorem{theorem}{Theorem}
 
\title{Optimal Universal Lossless Compression with Side Information}
\author{Yeohee~Im, ̃~\IEEEmembership{Student~Member, ̃~IEEE} and~Sergio~Verd\'u, ̃~\IEEEmembership{Fellow,~ ̃IEEE}
\thanks{This work was supported by the Center for Science of Information,
an NSF Science and Technology Center under Grant CCF-0939370. This paper was presented at the 2017 IEEE International Symposium on Information Theory \cite{IV2017}.}
\thanks{The authors are with the Department of Electrical Engineering, Princeton
University, Princeton, NJ 08544 USA (e-mail: yeoheei@princeton.edu; verdu@princeton.edu)}
}

\begin{document}
\maketitle

\begin{abstract}
This paper presents conditional versions of Lempel-Ziv (LZ) algorithm for settings where compressor and decompressor have access to the same side information.
We propose a fixed-length-parsing LZ algorithm with side information, motivated by the Willems algorithm, 
and prove the optimality for any stationary processes.
In addition, we suggest strategies to improve the algorithm which lower the data compression rate. 
A modification of a variable-length-parsing LZ algorithm with side information is proposed
and proved to be asymptotically optimal for any stationary and ergodic processes.
\end{abstract}


\section{Introduction}

An optimal lossless data compression to encode messages with the smallest possible average codelength was developed by Huffman \cite{H1952}.
Rissanen \cite{R1976} discovered arithmetic coding, which maps variable-length strings to binary intervals.
Since knowledge about the statistics of source is not available in many cases, universal source coding has received considerable interest. 
Lempel-Ziv (LZ) algorithms \cite{LZ1976,ZL1977,ZL1978} are widely known universal coding
that is implemented with parsing source, 
and underlie several prevalent compression applications.
The Context Tree Weighting method \cite{WST1995} weights the model distributions recursively using a binary context tree to calculate distribution and perform prediction. 

The LZ77 algorithm developed in \cite{ZL1977} parses strings into shorter phrases
by searching a maximal-length copy of the unparsed string, from the past string that has already conveyed.
The sliding-window version of LZ77 was shown to be optimal in \cite{WZ1994} for any stationary and ergodic processes.
Willems \cite{W1989} proposed a fixed-length counterpart of LZ coding,
and proved the data compression rate approaches the entropy rate capitalizing on Kac's lemma on repetition times \cite{K1947}.

In many data compression applications, compressor and decompressor have access to the same side information.
For example, the compressor may want to convey a new version of a file, 
whose previous version is available at both compressor and decompressor.
In bioinformatics, 
data compression with side information is useful in genomic sequencing,
where a known reference sequence can be regarded as side information \cite{YKK2001}.
Based on LZ77, \cite{PR2001} presented an algorithm that compresses a target sequence by parsing into phrases having a longest match with a reference.
In \cite{COM2012}, an algorithm was introduced to compress a target sequence given a reference genome with considering deletion and insertion operations, motivated by the sliding-window LZ coding.

File servers have to broadcast newly updated files to local users when software update is required, and in this case, old files can take the role of side information.
The \textit{rsync} algorithm \cite{TM1996} is designed to update a file 
with identifying identical parts between two files.
The first step of the \textit{rsync} algorithm calculates parity checksums of two files in order to locate where the files match.
For those matching segments, in the second step, a stronger checksum is compared.
On top of that, side information is present in data compression of image processing.
A noisy version of image serves as side information to compressor and decompressor in order to transmit the digitized image \cite{PR2001},
and a low-resolution image is given in the form of side information to decode higher-resolution image \cite{SK2000}.
In video compression, previous frames are used as side information to predict current frames \cite{AZG2002}.

There have been studies on universal data compression for settings where side information is given to compressor and decompressor.
In \cite{CKV2006}, an optimal algorithm with side information was designed 
based on the Context Tree Weighting method developed in \cite{WST1995}.
The conditional Multilevel Pattern Matching code was introduced in \cite{CKV2006},
which are optimal for any stationary and ergodic processes.
In addition, the problem of universal source coding with side information has been introduced applying LZ codings. 
Using the algorithm that parses joint strings \cite{Z1985} devised for universal decoding, 
the LZ78 algorithm was applied for side information in \cite{UK2003}.
In \cite{SB1995}, a conditional version of the sliding-window LZ algorithm is designed with side information.
The optimality of the algorithm was shown in \cite{JB2008} 
for stationary and ergodic sources with exponential rate for entropy, of which the probability of the complement of a typical set decreases exponentially.

This paper describes a universal fixed-length-parsing algorithm with side information, 
motivated by the algorithm without side information proposed by Willems in \cite{W1989}. 
Our proposed algorithm is shown to be optimal for any jointly stationary sequences.
Furthermore, we propose a modification of the sliding-window LZ algorithm with side information \cite{SB1995}, which allows variable-length-parsing,
and prove the optimality of the algorithm for any stationary and ergodic sources.
Section \ref{section:fix} and Section \ref{section:var} describe the algorithms,
which are analyzed in Section \ref{section:ana}.
The numerical results are given in Section \ref{section:sim}.

\section{Fixed-length-parsing universal compression} \label{section:fix}

In this section, we propose an extension of the Willems algorithm for fixed-length parsing which takes advantage of the side information available to the compressor and the decompressor.
Let $\mathbf{X}=(\cdots,X_1,X_2,\cdots)$ be source that will be encoded by the compressor.
Side information $\mathbf{Y}=(\cdots,Y_1,Y_2,\cdots)$ is known to both of the compressor and the decompressor.
$\mathbf{X}$ and $\mathbf{Y}$ take values on finite alphabets $\mathcal{A}$ and $\mathcal{B}$, respectively.
We assume the processes $\mathbf{X}$ and $\mathbf{Y}$ are jointly stationary but are not necessarily ergodic. 
As in the Willems algorithm, the proposed algorithm deals with the source and the side information by parsing them into phrases of fixed size $L$. 
The compressor aims to compress $N$ phrases of size $L$, i.e., $(X_1,\cdots,X_{NL})$.

\subsection{Algorithm 1}
Define $k\triangleq\lceil L\log_2|\mathcal{A}|\rceil$, 
which is the number of bits needed to represent $L$ symbols from the alphabet $\mathcal{A}$.
The algorithm starts by sending $k$ bits to convey the first phrase $(X_1,\cdots,X_L)$ without any compression.
At the $i$-th step, the compressor finds out where the same values of the $i$-th phrase, $(X_{(i-1)L+1},\cdots,X_{iL})$ and $(Y_{(i-1)L+1},\cdots,Y_{iL})$, appeared simultaneously before, for the last time.
Tracking back until the location of that $(\mathbf{X},\mathbf{Y})$-match, count the number of $\mathbf{Y}$-matches of the phrase $(Y_{(i-1)L+1},\cdots,Y_{iL})$.
As the compressor transmits the information of the number of $\mathbf{Y}$-matches, the decompressor can recognize the location of the $(\mathbf{X},\mathbf{Y})$-match uniquely.
The number of $\mathbf{Y}$-matches is conveyed using $h_k(\cdot)$ the universal prefix code for $[1 : 2^k]$\footnote{$[a:b]\triangleq\{t\in\mathbb{Z}~:~ a\leq t\leq b\}$.} in \cite{W1989}, the length of which is
	\begin{align}
	\ell(h_k(n))=\left\lbrace\begin{array}{ll} \lceil \log_2(1+k)\rceil+\lfloor \log_2n\rfloor,& \text{if }n<2^k,\\ \lceil \log_2(1+k)\rceil,&\text{if }n=2^k,\end{array}\right.
	\end{align}
where $\ell(\cdot)$ is the length function.

If the number of $\mathbf{Y}$-matches $n_i \in [1:2^k-1]$, the compressor sends $h_k(n_i)$.
When $n_i\geq 2^k$ or there is no match, $h_k(2^k)$ is transmitted, followed by the uncompressed sequence $(X_{(i-1)L+1},\cdots,X_{iL})$.

\begin{algorithm}[!t]
\caption{ }
\begin{algorithmic}
\State{Send $k$ bits to describe $(X_1,\cdots,X_L)$.}
\For{$i=2:N$}
	\begin{align*}
	s_i=\min\bigg\{t~:~ 0< t< (i-1)L+1,~(XY)_{(i-1)L+1}^{iL}=(XY)_{(i-1)L+1-t}^{iL-t}\bigg\}
	\end{align*}
\If{$s_i>0$} 
	\begin{align*}
	n_i=\left|\left\{t\in[1,s_i]~:~Y_{(i-1)L+1}^{iL}=Y_{(i-1)L+1-t}^{iL-t}\right\}\right|.	
	\end{align*}
\Else $~~n_i=0$.
\EndIf
\If{$n_i\in[1:2^k-1]$} send $h_k(n_i)$.
\Else 
\State{Send $h_k(2^k)$.}
\State{Send $X_{(i-1)L+1}^{iL}$ without compression using $k$ bits.}
\EndIf
\EndFor
\end{algorithmic}
\label{alg:main}
\end{algorithm}

\subsection{Algorithm 2} 
If Algorithm \ref{alg:main} cannot find an $(\mathbf{X},\mathbf{Y})$-match within the offset $2^k$, the compressor simply sends $h_k(2^k)$ 
and the uncompressed phrase.
There is, however, a fair chance of finding an $\mathbf{X}$-match in the past, 
permitting the compressor to capitalize on the original LZ77 algorithm.

After sending $i-1$ phrases, the compressor looks back for the $(\mathbf{X},\mathbf{Y})$-match of the $i$-th phrase,
and calculates the number of $\mathbf{Y}$-matches $n_i$ 
until the $(\mathbf{X},\mathbf{Y})$-match appears as in Algorithm 1. 
In case of $n_i\in[1:2^k]$, the flag bit 0 is sent 
to denote that an $(\mathbf{X},\mathbf{Y})$-match was found, and $h_k(n_i)$ is delivered.
Otherwise, the compressor forwards the flag bit 1 and finds out whether there exists an $\mathbf{X}$-match. 
When the match is found, the number of symbols, $r_i$, that it has to look back to reach the match is counted,
and $h_m(r_i)$ is sent, if $1\leq r_i\leq 2^m-1$, for some fixed integer parameter $m$. 
When $r_i\geq 2^m$ or no match was found, $h_m(2^m)$ is conveyed, followed by $k$ uncompressed bits.

\begin{algorithm}[!t]
\caption{ }
\begin{algorithmic}
\State{Send $k$ bits to describe $(X_1,\cdots,X_L)$.}
\For{$i=2:N$}
	\begin{align*}
	s_i=\min\bigg\{t~:~ 0< t< (i-1)L+1,~(XY)_{(i-1)L+1}^{iL}=(XY)_{(i-1)L+1-t}^{iL-t}\bigg\}
	\end{align*}
\If{$s_i>0$} 
	\begin{align*}
	n_i=\left|\left\{t\in[1,s_i]~:~Y_{(i-1)L+1}^{iL}=Y_{(i-1)L+1-t}^{iL-t}\right\}\right|.	
	\end{align*}
\Else $~~n_i=0$.
\EndIf
\If{$n_i\in[1:2^k]$}
\State{Send a flag bit 0 and $h_k(n_i)$.}
\Else 
\State{Send a flag bit 1.}
	\begin{align*}
	r_i=\min\bigg\{t~:~ 0< t< (i-1)L+1,~X_{(i-1)L+1}^{iL}=X_{(i-1)L+1-t}^{iL-t}\bigg\}
	\end{align*}
\If{$1\leq r_i\leq 2^m-1$} 
\State{send $h_m(r_i)$.}
\Else
\State{Send $h_m(2^m)$.}
\State{Send $X_{(i-1)L+1}^{iL}$ with $k$ bits.}
\EndIf
\EndIf
\EndFor
\end{algorithmic}
\label{alg:mod1}
\end{algorithm}

\subsection{Algorithm 3} \label{sub:modi2}

In Algorithm \ref{alg:main}, the prefix code $h_k(\cdot)$ is used to convey the number of $\mathbf{Y}$-matches found until the $(\mathbf{X},\mathbf{Y})$-match is reached. 
The code assigns a codeword when the number of matches is in $[1:2^k-1]$.
However, when the sequences are not long enough, the number of all $\mathbf{Y}$-matches can be smaller than $2^k-1$, in which case it is inefficient to consider a code for $[1:2^k-1]$. 
The algorithm is modified in order that the compressor can reduce the amount of bits to send, by adopting a prefix code with another parameter.

At the $i$-th step, the match of the $i$-th phrase in the past is found, and the number of $\mathbf{Y}$-matches $n_i$ until the match is counted as before.
Besides, the algorithm computes $p_i$, the number of $\mathbf{Y}$-matches from the beginning until the current time. 
If $p_i<2^k-1$, then there is no need to use the code $h_k(\cdot)$, since it is efficient only when every index $1,2,\cdots,2^k-1$ is a possible candidate to be encoded.
Choose the parameter $\bar{k}_i=\lceil\log_2(p_i+1)\rceil$. 
The decompressor as well as the compressor can compute $\bar{k}_i$, since they are aware of every location of the $\mathbf{Y}$-matches.
The compressor conveys $h_{\bar{k}_i}(n_i)$ if $n_i\in[1:2^{\bar{k}_i}-1]$, and otherwise, it sends $h_{\bar{k}_i}(2^{\bar{k}_i})$ with the uncompressed sequence. 
After several phrases have been sent, $p_i$ will become as large as $2^k-1$.
In that case, $h_k(2^k)$ can be used.
This modification is useful for the original Willems algorithm without side information \cite{W1989} as well.

\begin{algorithm}[!t]
\caption{ }
\begin{algorithmic}
\State{Send $k$ bits to describe $(X_1,\cdots,X_L)$.}
\For{$i=2:N$}
	\begin{align*}
	s_i=\min\bigg\{t~:~ 0< t< (i-1)L+1,~(XY)_{(i-1)L+1}^{iL}=(XY)_{(i-1)L+1-t}^{iL-t}\bigg\}
	\end{align*}
\If{$s_i>0$} 
	\begin{align*}
	n_i=\left|\left\{t\in[1,s_i]~:~Y_{(i-1)L+1}^{iL}=Y_{(i-1)L+1-t}^{iL-t}\right\}\right|.	
	\end{align*}
\Else $~~n_i=0$.
\EndIf
	\begin{align*}
	p_i=\left|\left\{t\in[1,(i-1)L]~:~Y_{(i-1)L+1}^{iL}=Y_{(i-1)L+1-t}^{iL-t}\right\}\right|.	
	\end{align*}
\If{$p_i<2^k-1$} $\bar{k}_i=\lceil\log_2(p_i+1)\rceil$.
\If{$n_i\in[1:2^{\bar{k}_i}-1]$} send $h_{\bar{k}_i}(n_i)$.
\Else 
\State{Send $h_{\bar{k}_i}(2^{\bar{k}_i})$.}
\State{Send $X_{(i-1)L+1}^{iL}$ using $k$ bits.}
\EndIf
\Else
\If{$n_i\in[1:2^k-1]$} send $h_k(n_i)$.
\Else 
\State{Send $h_k(2^k)$.}
\State{Send $X_{(i-1)L+1}^{iL}$ using $k$ bits.}
\EndIf
\EndIf
\EndFor
\end{algorithmic}
\label{alg:mod2}
\end{algorithm}

\section{Variable-length-parsing universal compression}\label{section:var}

An extension of the LZ77 algorithm for side information was proposed in \cite{SB1995}, 
and a modified algorithm is described in this section.
We assume the processes $\mathbf{X}$ and $\mathbf{Y}$ are jointly stationary and ergodic, and 
the number of symbols the compressor aims to convey is denoted by $K$.

\subsection{Algorithm 4}

As in the sliding-window Lempel Ziv algorithm \cite{ZL1977} and the conditional version \cite{SB1995}, the compressor and the decompressor seek for a match within a sliding window. 
Let $n_w$ be the window size. 

First, the compressor sends $(X_1,\cdots,X_{n_w})$ using $\lceil n_w\log_2|\mathcal{A}|\rceil$ bits.
Let $u_1=n_w+1$ denotes the location for the compressor to begin with.
At the $i$-th step, the compressor finds the longest $(\mathbf{X},\mathbf{Y})$-match within the window,
and conveys the length of the longest match, $l_i$, to the decompressor, 
using the prefix code $g(\cdot)$ for nonnegative integers introduced in \cite[Appendix]{WZ1994}, 
the length of which is bounded as
	\begin{align}
	\ell(g(n))\leq \gamma \log_2 (n+1)\label{eqn:g}
	\end{align}
for some constant $\gamma$.
At this point, both of the compressor and the decompressor can compute how many $\mathbf{Y}$-matches of length $l_i$ the window has, which is denoted by $c_i$.
If $\lceil \log_2 c_i\rceil$ is longer than $\lceil l_i \log_2 |\mathcal{A}|\rceil$, 
then the compressor sends $\lceil l_i \log_2 |\mathcal{A}|\rceil$ bits to describe $(X_{u_i},\cdots,X_{u_i+l_i-1})$.
Otherwise, the location of the $(\mathbf{X},\mathbf{Y})$-match is transmitted to the decompressor using $\lceil \log_2 c_i\rceil$ bits.
In this case, since the decompressor can identify every $\mathbf{Y}$-match in the window, the $(\mathbf{X},\mathbf{Y})$-match is also can be found.
Set $u_{i+1}=u_i+l_i$ and keep parsing until it reaches the end of the sequences.

Note that while the algorithm presented in \cite{SB1995} counts the number of $\mathbf{Y}$-matches between the current position and the longest $(\mathbf{X},\mathbf{Y})$-match, and encodes the number with a prefix code for integers,
the algorithm in this section counts the number of all $\mathbf{Y}$-matches in the window.

\begin{algorithm}[!t]
\caption{ }
\begin{algorithmic}
\State{Send $\lceil n_w\log_2|\mathcal{A}|\rceil$ bits to describe $(X_1,\cdots,X_{n_w})$.}
\State{$i=1$.}
\State{$u_i=n_w+1$.}
\While{$u_i\leq K$}
	\begin{align*}
	l_i=\max\left\{n\leq K-u_i+1~:~\exists~t\in[1:n_w] ~\text{ s.t. }(XY)_{u_i}^{u_i+n-1}=(XY)_{u_i-t}^{u_i+n-1-t}\right\}	
	\end{align*}
\If{$l_i=0$} $l_i=1$.\EndIf
\State{Send $g(l_i)$.}
	\begin{align*}
	c_i=\left|\left\{ t\in[1:n_w]~:~Y_{u_i}^{u_i+l_i-1}=Y_{u_i-t}^{u_i+l_i-1-t}\right\}\right|
	\end{align*}

\If{$\lceil\log_2 c_i\rceil\geq \lceil l_i\log_2|\mathcal{A}|\rceil$ or $l_i=1$}
\State{Send $X_{u_i}^{u_i+l_i-1}$ using $\lceil l_i\log_2|\mathcal{A}|\rceil$ bits.}
\Else ~Send the location of $(\mathbf{X},\mathbf{Y})$-match using $\lceil \log_2 c_i\rceil$ bits.
\EndIf
\State{$u_{i+1}=u_i+l_i$.}
\State{$i=i+1$.}
\EndWhile
\end{algorithmic}
\label{alg:var}
\end{algorithm}

\section{Analysis} \label{section:ana}

In this section we show the optimality of Algorithm \ref{alg:main} for stationary processes and compare the performance of the modified algorithms -- Algorithm \ref{alg:mod1} and \ref{alg:mod2} -- with that of Algorithm \ref{alg:main}.
Further, we prove the optimality of Algorithm \ref{alg:var} for stationary and ergodic processes.

For the fixed-length-parsing algorithms, let $w_1^i(X_1^{iL}|Y_1^{iL})$, $w_2^i(X_1^{iL}|Y_1^{iL})$, $w_3^i(X_1^{iL}|Y_1^{iL})$ be the codewords constructed by Algorithms \ref{alg:main}, \ref{alg:mod1}, and \ref{alg:mod2}, respectively, to send the $i$-th phrase $X_{(i-1)L+1}^{iL}$, when the previous phrases are $X_1^{(i-1)L}$ and side information is $Y_1^{iL}$.
The entire codeword for a source $\mathbf{a}$ with side information $\mathbf{b}$ is denoted by $w_1(\mathbf{a}|\mathbf{b}),\cdots,w_4(\mathbf{a}|\mathbf{b})$.

\subsection{Optimality of Algorithm \ref{alg:main}}

Willems \cite{W1989} showed the asymptotic optimality of his fixed-length-parsing algorithm in the sense of expected length
for stationary sources. Here we show that the length per symbol of the algorithm with side information approaches the conditional entropy rate as the blocklength and the number of phrases go to infinity.
For any doubly infinite sequence $\mathbf{a}$$=(\cdots,a_{-1},a_0,a_1,\cdots)$, the repeated recurrence times are defined as 
	\begin{align*}
	T_{L,0}(\mathbf{a})&=0,\\
	T_{L,j}(\mathbf{a})&=\min\left\{t\in\mathbb{N}:t>T_{L,j-1}(\mathbf{a}),a_1^L=a_{1-t}^{L-t}\right\}~\forall j\geq 1.
	\end{align*}
A random variable $C$ indicates the number of $\mathbf{Y}$-matches that appeared in the past since the latest $(\mathbf{X},\mathbf{Y})$-match, 
which is 
	\begin{align}
	C=\left|\left\{t\in\mathbb{N}~:~t\leq T_{L,1}(\mathbf{X},\mathbf{Y}),~Y_1^L=Y_{1-t}^{L-t}\right\}\right|.\label{eqn:C}
	\end{align}
The following lemma provides the relationship between the expected value of $C$ and the conditional probability. 

\begin{lemma}\label{lemma:kac}
For any stationary processes $\mathbf{X},\mathbf{Y}$ and sequences $\mathbf{x},\mathbf{y}$, 
the conditional probability of $\mathbf{x}$ given $\mathbf{y}$ and the conditional expectation of $C$ satisfy 
\begin{align}
\mathbb{E}\left[C|X_1^L=\mathbf{x},Y_1^L=\mathbf{y}\right]\leq\left(\mathbb{P}\left[X_1^L=\mathbf{x}|Y_1^L=\mathbf{y}\right]\right)^{-1},
\end{align}
where $C$ is defined in \eqref{eqn:C}.
\end{lemma}
\begin{proof}
We simplify notation as 
\begin{align}
T_j\triangleq T_{L,j}(\mathbf{Y}).	
\end{align}
Define a probability 
	\begin{align}
	q_n \triangleq \mathbb{P}\left[ X_{1-T_j}^{L-T_j}\neq \mathbf{x}~\forall j \in[1: n] 	~|~ Y_1^L=\mathbf{y}\right].
	\end{align}

Note that the difference between $q_n$ and $q_{n+1}$ 
represents the probability that $\mathbf{X}$-match appears for the first time with the interval $T_{n+1}$, 
among the locations of $\mathbf{Y}$-matches, as
	\begin{align}
	q_n-q_{n+1}=\mathbb{P}\left[X_{1-T_j}^{L-T_j}\neq	\mathbf{x}~\forall j\in[1:n],~X_{1-T_{n+1}}^{L-T_{n+1}}=\mathbf{x}	~|~ Y_1^L=\mathbf{y}\right].
	\end{align}
This probability can be interpreted in the following way by stationarity as well:
	\begin{align*}
	&q_n-q_{n+1}\\
	=&\mathbb{P}\left[X_{1-T_j}^{L-T_j}\neq	\mathbf{x}~\forall~ j\in[1:n],~X_{1-T_{n+1}}^{L-T_{n+1}}=\mathbf{x}	~|~ Y_1^L=\mathbf{y}\right]\\
	=&\sum_{(t_1,\cdots,t_{n+1})} \mathbb{P}\left[T_1^{n+1}=t_1^{n+1}|Y_1^L=\mathbf{y}\right]\mathbb{P}\left[X_{1-t_j}^{L-t_j}\neq \mathbf{x}~\forall j\in[1:n],~X_{1-t_{n+1}}^{L-t_{n+1}}=\mathbf{x}|Y_1^L=\mathbf{y}	,T_1^{n+1}=t_1^{n+1}\right]\\
	=&\sum_{(t_1,\cdots,t_{n+1})} \mathbb{P}\left[E|Y_1^L=\mathbf{y}\right]\mathbb{P}\left[ X_{1-t_j+t_1}^{L-t_j+t_1}\neq \mathbf{x}~\forall j\leq n-1, X_{1-t_{n+1}+t_1}^{L-t_{n+1}+t_1}=\mathbf{x}|E\right]\\
	=&\mathbb{P}\left[X_1^L\neq \mathbf{x},X_{1-T_j}^{L-T_j}\neq \mathbf{x}\forall j\leq n-1,X_{1-T_{n}}^{L-T_{n}}=\mathbf{x}|Y_1^L=\mathbf{y}\right],
	\end{align*}
where the event $E$ is defined as
	\begin{align}
	E&=\bigg\{ Y_1^L=\mathbf{y},~T_j=1-t_{j+1}+t_1~\forall~j\in[1:n],~Y_{1+t_1}^{L+t_1}=\mathbf{y},~Y_t^{t+L-1}\neq\mathbf{y}~\forall t\in[2:t_1]\bigg\}.	
	\end{align}

That is to say, $q_n-q_{n+1}$ is the probability that the source matches for the first time at the $(n+1)$-th recurrence time in the past, and at the same time it is the probability that the source does not match at the current time and matches for the first time at the $n$-th recurrence time.

The conditional probability of $C$ and $X_1^L$ becomes 
	\begin{align*}
	&\mathbb{P}\left[ C=n,~X_1^L=\mathbf{x}~|~Y_1^L=\mathbf{y}\right]\\
	&=\mathbb{P}\left[X_1^L=\mathbf{x},~X_{1-T_j}^{L-T_j}\neq	\mathbf{x}~\forall j\leq n-1,~X_{1-T_{n}}^{L-T_n}=\mathbf{x}	~|~ Y_1^L=\mathbf{y}	\right]\\
	&=\mathbb{P}\left[X_{1-T_j}^{L-T_j}\neq	\mathbf{x}~\forall~ j\leq n-1,~X_{1-T_n}^{L-T_n}=\mathbf{x}	~|~ Y_1^L=\mathbf{y}\right]
	-\mathbb{P}\left[X_1^L\neq \mathbf{x},X_{1-T_j}^{L-T_j}\neq \mathbf{x}\forall j\leq n-1,X_{1-T_{n}}^{L-T_{n}}=\mathbf{x}|Y_1^L=\mathbf{y}\right]\\
	&=	(q_{n-1}-q_n)-(q_n-q_{n+1}).
	\end{align*}
Finally, we get the desired result from
	\begin{align}
	&\mathbb{E}\left[C|X_1^L=\mathbf{x},Y_1^L=\mathbf{y}	\right]	\mathbb{P}\left[X_1^L=\mathbf{x}~|~Y_1^L=\mathbf{y}\right]\nonumber\\
	&=\lim_{n\rightarrow\infty}\sum_{j=1}^n j\left(	q_{j-1}-2q_j+q_{j+1}\right)\\
	&=\lim_{n\rightarrow\infty} \left[	1-q_n-n(q_n-q_{n+1})	\right]\\
	&=\lim_{n\rightarrow \infty} (1-q_n)\label{eqn:o}\\
	&=\mathbb{P}\left[\bigcup_{j=1}^\infty \left\{ 	X_{1-T_{j}}^{L-T_{j}}=\mathbf{x}\right\}~|~ Y_1^L=\mathbf{y}	\right]\leq 1,
	\end{align}
where \eqref{eqn:o} holds because 
	$
	\sum_{j=1}^\infty (q_j-q_{j+1} )<\infty,
	$
and accordingly, $q_n-q_{n+1}=o\left(\frac{1}{n}\right)$.
\end{proof}
Note that Lemma \ref{lemma:kac} reduces to Kac's lemma \cite{K1947} in the case without side information, which implies that the mean-recurrence time of a sequence is lower than the reciprocal of the probability of the sequence.

Based on Lemma \ref{lemma:kac}, the average codelength of the algorithm with side information is analyzed next.
From this point, we abbreviate $T\triangleq T_{L,1}(\mathbf{X},\mathbf{Y})$.

\begin{theorem}\label{theorem:block}
The conditional average codelength of the $i$-th phrase is bounded as
\begin{align}
&\mathbb{E}\left[\ell\left(w_1^i(X_1^{iL}|Y_1^{iL})\right)|X_{(i-1)L+1}^{iL}=\mathbf{x},Y_{(i-1)L+1}^{iL}=\mathbf{y}\right]\nonumber\\
&\leq \imath_{X_1^L|Y_1^L}(\mathbf{x}|\mathbf{y})+\lceil\log_2(1+k)\rceil +\frac{k}{((i-1)L+1)\mathbb{P}[X_1^L=\mathbf{x},Y_1^L=\mathbf{y}]},
\end{align}
for any $\mathbf{x}\in\mathcal{A}^L,~\mathbf{y}\in\mathcal{B}^L$.
\end{theorem}
\begin{proof}
Define the event $E_{\mathbf{x}\mathbf{y}}\triangleq 	\left\{ X_1^L=\mathbf{x},Y_1^L=\mathbf{y}\right\}.$
Using stationarity, the expected length of codeword of the $i$-th phrase is 
	\begin{align*}
	&\mathbb{E}\left[\ell\left(w_1^i(X_1^{iL}|Y_1^{iL})\right)|X_{(i-1)L+1}^{iL}=\mathbf{x},Y_{(i-1)L+1}^{iL}=\mathbf{y}\right]\\
	&=\mathbb{E}\left[\ell\left(w_1^i(X_{1-(i-1)L}^{L}|Y_{1-(i-1)L}^{L})\right)
	|E_{\mathbf{x}\mathbf{y}}\right]\\
	&=\sum_{j=1}^{2^k-1}\lfloor \log_2j\rfloor \mathbb{P}\left[ 	C=j,T\leq (i-1)L |E_{\mathbf{x}\mathbf{y}}	\right]+k\mathbb{P}\left[		C\geq 2^k \text{ or } T\geq (i-1)L+1|E_{\mathbf{x}\mathbf{y}}\right]+\lceil\log_2(1+k)\rceil\\
	&\leq \sum_{j=1}^{2^k-1}\lfloor \log_2j\rfloor~ \mathbb{P}\left[C=j|E_{\mathbf{x}\mathbf{y}}\right]+\sum_{j=2^k}^\infty \log_2j ~\mathbb{P}\left[		C=j ~|E_{\mathbf{x}\mathbf{y}}\right]
	+k \mathbb{P}\left[		C\leq 2^k-1,~ T\geq (i-1)L+1|E_{\mathbf{x}\mathbf{y}}\right]+\lceil\log_2(1+k)\rceil\\
	&\stackrel{(a)}{\leq}
	 \log_2\mathbb{E}\left[C|E_{\mathbf{x}\mathbf{y}}\right]+\lceil\log_2(1+k)\rceil
	 +k\sum_{j=1}^{2^k-1}\mathbb{P}\left[C=j|E_{\mathbf{x}\mathbf{y}}\right] \mathbb{P}\left[ T_j\geq (i-1)L+1|C=j,E_{\mathbf{x}\mathbf{y}}\right]\\
	&\stackrel{(b)}{\leq} 	\imath_{X_1^L|Y_1^L}(\mathbf{x}|\mathbf{y})+\lceil\log_2(1+k)\rceil +\frac{k\mathbb{E}[C|E_{\mathbf{x}\mathbf{y}}]}{((i-1)L+1)\mathbb{P}[Y_1^L=y]}\\
	&\stackrel{(c)}{\leq} \imath_{X_1^L|Y_1^L}(\mathbf{x}|\mathbf{y})+\lceil\log_2(1+k)\rceil+\frac{k}{((i-1)L+1)\mathbb{P}[E_{\mathbf{x}\mathbf{y}}]},
	\end{align*}
where $(a)$ holds by the concavity of logarithm and Jensen's inequality, and $(b)$ follows by Lemma \ref{lemma:kac}, Markov's inequality and 
	\begin{align}
	&\mathbb{E}\left[T_j|Y_1^L=\mathbf{y}\right]=j\, \mathbb{E}[T_1|Y_1^L=\mathbf{y}]\label{eqn:j}\\
	&\leq \frac{j}{\mathbb{P}[Y_1^L=\mathbf{y}]}.
	\end{align}
\eqref{eqn:j} is given by
\begin{align}
&\mathbb{P}\left[T_n-T_{n-1}=t, Y_1^L=\mathbf{y}\right]\nonumber\\
=&\sum_{t_1^{n-1} : 0<t_1<\cdots<t_{n-1}}	\mathbb{P}\left[Y_1^L=\mathbf{y}, T_1^{n-1}=t_1^{n-1}\right]\mathbb{P}\left[T_n-T_{n-1}=t| Y_1^L=\mathbf{y}, T_1^{n-1}=t_1^{n-1}\right]\\
=&\sum_{t_1^{n-1}} \mathbb{P}\big[Y_{a+1}^{a+L}=\mathbf{y} ~\forall a\in \mathcal{S}, Y_{b+1}^{b+L}\neq \mathbf{y}~\forall b\in [-t_{n-1}:0]\backslash \mathcal{S}\big]\nonumber\\
&\times \mathbb{P}\left[Y_{1-t_{n-1}-t}^{L-t_{n-1}-t}=\mathbf{y}| Y_{a+1}^{a+L}=\mathbf{y} ~\forall a\in \mathcal{S}, ~Y_{b+1}^{b+L}\neq \mathbf{y}~\forall b\in [-t_{n-1}:0]\backslash \mathcal{S}\right]\\
=&\sum_{t_1^{n-1}} \mathbb{P}\left[Y_{a+1+t_{n-1}}^{a+L+t_{n-1}}=\mathbf{y} ~\forall a\in \mathcal{S},~Y_{b+1+t_{n-1}}^{b+L+t_{n-1}}\neq \mathbf{y}~\forall b\in [-t_{n-1}:0]\backslash \mathcal{S}\right]\nonumber\\
&\times \mathbb{P}\left[Y_{1-t}^{L-t}=\mathbf{y}| Y_{a+1+t_{n-1}}^{a+L+t_{n-1}}=\mathbf{y} ~\forall a\in \mathcal{S},~Y_{b+1+t_{n-1}}^{b+L+t_{n-1}}\neq \mathbf{y}~\forall b\in [-t_{n-1}:0]\backslash \mathcal{S}\right]\\
=&\sum_{t_1^{n-1}} \mathbb{P}\left[Y_{1-t}^{L-t}=\mathbf{y}, Y_{a+1+t_{n-1}}^{a+L+t_{n-1}}=\mathbf{y} ~\forall a\in \mathcal{S},~Y_{b+1+t_{n-1}}^{b+L+t_{n-1}}\neq \mathbf{y}~\forall b\in [-t_{n-1}:0]\backslash \mathcal{S}\right]\\
=&\mathbb{P}\left[Y_{1-t}^{L-t}=\mathbf{y},Y_1^L=\mathbf{y}\right]\\
=&\mathbb{P}\left[T_1=t,Y_1^L=\mathbf{y}\right]~\forall~ n,t\in\mathbb{N},~\mathbf{y}\in\mathcal{Y}^L,
\end{align}
where a set $\mathcal{S}$ is defined as $\mathcal{S}=\{0\}\cup \{-t_1,\cdots,-t_{n-1}\}$.
Inequality $(c)$ is implied by Lemma \ref{lemma:kac}.
\end{proof}

Utilizing the bound for the codelength of each phrase given in Theorem \ref{theorem:block}, we show that the average compression rate is asymptotically bounded by the conditional entropy rate.

\begin{theorem}\label{theorem:optimal}
Algorithm \ref{alg:main} is asymptotically optimal in the sense that
\begin{align}
\lim_{L\rightarrow \infty}\lim_{N\rightarrow \infty}\frac{\mathbb{E}\left[\ell(w_1(X_1^{NL}|Y_1^{NL}))	\right]}{NL}\leq H(\mathbf{X}|\mathbf{Y}).\label{eqn:infinite}
\end{align}
\end{theorem}
\begin{proof}
For any finite $N$ and $L$, the expected codelength is bounded as 
\begin{align}
	&\mathbb{E}\left[\ell(w_1(X_1^{NL}|Y_1^{NL}))	\right]\nonumber\\
	&=k+\sum_{i=2}^N \mathbb{E}\left[l\left(w_1^i(X_1^{iL}|Y_1^{iL})\right)	\right]\\
	&=k+\sum_{i=2}^N\sum_{\mathbf{x}\in\mathcal{A}^L}\sum_{\mathbf{y}\in\mathcal{B}^L}\mathbb{P}\left[	X_{(i-1)L+1}^{iL}=\mathbf{x},Y_{(i-1)L+1}^{iL}=\mathbf{y}\right]\mathbb{E}\left[\ell\left(w_1^i(X_1^{iL}|Y_1^{iL})\right)|X_{(i-1)L+1}^{iL}=\mathbf{x},Y_{(i-1)L+1}^{iL}=\mathbf{y}\right]\\
	&\leq \sum_{i=2}^N\sum_{\mathbf{x}\in\mathcal{A}^L}\sum_{\mathbf{y}\in\mathcal{B}^L}\mathbb{P}\left[E_{\mathbf{x}\mathbf{y}}\right]	\bigg(\lceil \log_2(1+k)\rceil +\imath_{X_1^L|Y_1^L}(\mathbf{x}|\mathbf{y})+\frac{k}{\left((i-1)L+1\right) \mathbb{P}\left[E_{\mathbf{x}\mathbf{y}}\right]}	\bigg)+k
	\\&=k+(N-1)\lceil \log_2(1+k)\rceil+(N-1)H(X_1^L|Y_1^L)+k|\mathcal{A}|^L|\mathcal{B}|^L\sum_{i=2}^N\frac{1}{(i-1)L+1}.
	\end{align}
As the number of phrases goes to infinity, the asymptotic compression rate can be upper bounded as:
	\begin{align}
	&\lim_{N\rightarrow \infty}\frac{\mathbb{E}\left[\ell(w_1(X_1^{NL}|Y_1^{NL}))	\right]}{NL}\nonumber\\
	 &\leq\frac{\lceil \log_2(1+k)\rceil}{L}+\frac{H(X_1^L|Y_1^L)}{L}+\frac{k|\mathcal{A}|^L|\mathcal{B}|^L}{L}\lim_{N\rightarrow\infty}\frac{1}{N}\sum_{i=2}^N\frac{1}{(i-1)L+1}\\
	&=\frac{\lceil \log_2(1+k)\rceil}{L}+\frac{H(X_1^L|Y_1^L)}{L}.\label{eqn:limit}
	\end{align}
Therefore, \eqref{eqn:infinite} follows as $L$ goes to infinity.
\end{proof}

\subsection{Algorithm \ref{alg:mod1}}

While Algorithm \ref{alg:main} exploits $(\mathbf{X},\mathbf{Y})$-matches only, $\mathbf{X}$-matches are also used in Algorithm \ref{alg:mod1} in order to denote the match location.
However, this modified strategy necessitates a flag bit to indicate whether the codeword signals an $(\mathbf{X},\mathbf{Y})$-match or an $\mathbf{X}$-match.
We show a sufficient condition for Algorithm \ref{alg:mod1} to outperform Algorithm \ref{alg:main}.

\begin{table*}[!t]
\captionof{table}{Comparison of Codelengths}\label{table:code}
\centering\small
\begin{tabular}{llll}
& Cases& Algorithm \ref{alg:main} & Algorithm \ref{alg:mod1} \\
 1)& $1\leq n_i\leq 2^k-1$ & $\lceil\log_2(1+k)\rceil+\lfloor \log_2n_i\rfloor$ &$\lceil\log_2(1+k)\rceil+\lfloor \log_2n_i\rfloor+1$ \\
2)& $n_i=2^k$& $\lceil \log_2(1+k)\rceil+k$ & $\lceil \log_2(1+k)\rceil+1$ \\
 3)& $n_i\geq 2^k+1,~r_i\leq 2^{m-1}$& $\lceil \log_2(1+k)\rceil+k$& $\lceil \log_2(1+m)\rceil+\lfloor \log_2r_i\rfloor+1$\\
 4)& $n_i\geq 2^k+1,~r_i\geq 2^m$ &$\lceil \log_2(1+k)\rceil+k$ &$\lceil \log_2(1+m)\rceil+k+1$
\end{tabular}
\end{table*}

The codelengths are summarized in Table \ref{table:code}, where $n_i$ is the number of $\mathbf{Y}$-matches until the $(\mathbf{X},\mathbf{Y})$-match is reached, and $r_i$ is the number of symbols needed to reach the $\mathbf{X}$-match, as before.
For cases $2),3),4)$, Algorithm \ref{alg:mod1} constructs a shorter codeword than Algorithm \ref{alg:main} does if $m<k$, while Algorithm \ref{alg:mod1} outputs one more bit for case $1)$.
Hence, Algorithm \ref{alg:mod1} is efficient when the frequency of case $1)$ is low.
Denote
	\begin{align*}
	E_1&\triangleq \left\{C\leq 2^k-1,~	T\leq (i-1)L\right\},\\
	E_2&\triangleq \left\{C=2^k,~ T\leq (i-1)L\right\},\\
	E_{3,t}&\triangleq \left(\left\{C\geq 2^k+1\right\}\cup \left\{T\geq(i-1)L+1\right\}\right)\cap\left\{T_1(\mathbf{X})=t\right\},\\
	E_4&\triangleq \left(\left\{C\geq 2^k+1\right\}\cup \left\{T\geq(i-1)L+1\right\}\right)\cap \left\{T_1(\mathbf{X})>d_i\right\},
	\end{align*}
where $d_i\triangleq \min(2^m-1,(i-1)L)$. The difference between the average codelengths is
	\begin{align}
	&\mathbb{E}\left[	\ell(w_1^i(X_1^{iL}|Y_1^{iL}))|X_{(i-1)L+1}^{iL}=\mathbf{x},Y_{(i-1)L+1}^{iL} =\mathbf{y}\right]
	-\mathbb{E}\left[	\ell(w_2^i(X_1^{iL}|Y_1^{iL}))|X_{(i-1)L+1}^{iL}=\mathbf{x},Y_{(i-1)L+1}^{iL} =\mathbf{y}\right]\nonumber\\
	&=-\mathbb{P}\left[E_1|E_{\mathbf{x}\mathbf{y}}\right]+(k-1)\mathbb{P}\left[E_2|E_{\mathbf{x}\mathbf{y}}\right] +\sum_{t=1}^{d_i}\mathbb{P}\left[E_{3,t}|E_{\mathbf{x}\mathbf{y}}\right]
	\left(\lceil \log_2(1+k)\rceil+k-\lceil\log_2(1+m)\rceil-\lfloor\log_2 t\rfloor-1	\right) \nonumber\\
	&+\mathbb{P}[E_4|E_{\mathbf{x}\mathbf{y}}] \left(\lceil \log_2(1+k)\rceil-\lceil\log_2(1+m)\rceil-1	\right)\\
	&\geq -\mathbb{P}\left[E_1|E_{\mathbf{x}\mathbf{y}}\right]+(k-1)\mathbb{P}\left[	E_2|E_{\mathbf{x}\mathbf{y}}\right]+\mathbb{P}\left[\cup_{t\leq d_i} E_{3,t}|E_{\mathbf{x}\mathbf{y}}\right]
	\left(\lceil \log_2(1+k)\rceil+k-\lceil\log_2(1+m)\rceil-m-1	\right)\nonumber\\
	&+\mathbb{P}\big[E_4|E_{\mathbf{x}\mathbf{y}}\big] \left(\lceil \log_2(1+k)\rceil-\lceil\log_2(1+m)\rceil-1	\right)\\
	&\geq -\mathbb{P}\left[E_1|E_{\mathbf{x}\mathbf{y}}\right]+
	\mathbb{P}\left[E_1^c|E_{\mathbf{x}\mathbf{y}}\right]\left(\lceil \log_2(1+k)\rceil-\lceil\log_2(1+m)\rceil-1	\right).
	\end{align}
This gap is positive if
	\begin{align}
	\mathbb{P}\left[E_1|E_{\mathbf{x}\mathbf{y}}	\right]\leq 1-\frac{1}{\lceil \log_2(1+k)\rceil-\lceil\log_2(1+m)\rceil},\label{eqn:suff}
	\end{align}
which means that Algorithm \ref{alg:mod1} is advantageous for relatively short sequences.

\subsection{Algorithm \ref{alg:mod2}}

Described in Section \ref{sub:modi2}, Algorithm \ref{alg:mod2} improves performance by exploiting the fact
that the decompressor is able to locate every match of the side information.
Let $\hat{\mathcal{B}}\subset\mathcal{B}^{iL}$ be a set composed of $\mathbf{y}\in\mathcal{B}^{iL}$ such that the number of matches, $p_i$, of the $i$-th phrase among the $i-1$ previous phrases is smaller than $2^k-1$. 
If $\mathbf{y}\in\hat{\mathcal{B}}$, a new parameter $\bar{k}_i=\log_2(p_i+1)$ replaces $k$. The conditional average codelength of Algorithm 3 for the $i$-th phrase is
	\begin{align}
	&\mathbb{E}\left[\ell(w_3^i(X_1^{iL}|Y_1^{iL}))|X_{(i-1)L+1}^{iL}=\mathbf{x},Y_1^{iL} =\mathbf{y}\right]\nonumber\\
	&=\sum_{j=1}^{2^{\bar{k}_i}-1} \lfloor\log_2 j\rfloor
	 \mathbb{P}\left[ C=j,T\leq (i-1)L|X_1^L=\mathbf{x},Y_{1-(i-1)L}^{L} =\mathbf{y}\right]\nonumber\\
	&+k\mathbb{P}\left[ C\geq 2^{\bar{k}_i}\text{ or }T>(i-1)L|X_1^L=\mathbf{x},Y_{1-(i-1)L}^{L} =\mathbf{y}\right]+\lceil\log_2 (1+\bar{k}_i)\rceil,\quad\forall ~\mathbf{y}\in\hat{\mathcal{B}}.\label{eqn:mod2}
	\end{align}
The second term of of \eqref{eqn:mod2} is
	\begin{align}
		&\mathbb{P}\left[ C\geq 2^{\bar{k}_i}\text{ or }T>(i-1)L|X_1^L=\mathbf{x},Y_{1-(i-1)L}^{L} =\mathbf{y}\right]=\mathbb{P}\left[T>(i-1)L|X_1^L=\mathbf{x},Y_{1-(i-1)L}^{L} =\mathbf{y}\right],
	\end{align}
since 
	\begin{align}
	T\leq (i-1)L\quad\implies \quad C\leq 2^{\bar{k}_i}-1.	
	\end{align}
The performance gap corresponding to the $i$-th phrase between Algorithm \ref{alg:main} and Algorithm 3 for this case results in 
\begin{align}
&\mathbb{E}\left[	\ell(w_1^i(X_1^{iL}|Y_1^{iL}))|X_{(i-1)L+1}^{iL}=\mathbf{x},Y_1^{iL} =\mathbf{y}\right]-\mathbb{E}\left[	\ell(w_3^i(X_1^{iL}|Y_1^{iL}))|X_{(i-1)L+1}^{iL}=\mathbf{x},Y_1^{iL} =\mathbf{y}\right]\nonumber\\
&=\sum_{j=2^{\bar{k}_i}}^{2^k-1} \lfloor \log_2j\rfloor
\mathbb{P}\left[C=j,T\leq(i-1)L|X_1^L=\mathbf{x},Y_{1-(i-1)L}^L=\mathbf{y}\right]
+\lceil\log_2 (1+k)\rceil-\lceil\log_2 (1+\bar{k}_i)\rceil\\
&=\lceil\log_2 (1+k)\rceil-\lceil\log_2 (1+\bar{k}_i)\rceil\geq 0,~\forall~\mathbf{y}\in\hat{\mathcal{B}}.
\end{align}
Hence, when $\mathbf{y}\in\hat{\mathcal{B}}$, the performance of Algorithm 3 is at least as good as that of Algorithm \ref{alg:main}, while if $\mathbf{y}\not\in\hat{\mathcal{B}}$, both algorithms yield the same expected codelength.
Thus, Algorithm \ref{alg:main} always constructs codewords no shorter than those of Algorithm 3, which exploits the fact that the first several phrases typically find few side information matches.

\subsection{Optimality of Algorithm \ref{alg:var}}

In \cite{WZ1994}, the sliding-window LZ algorithm was shown to be optimal.
We prove in this subsection that Algorithm \ref{alg:var} is optimal,
as the data compression rate approaches the conditional entropy rate when the window size and the source length increase.

\begin{theorem}\label{theorem:window}
Algorithm \ref{alg:var} is asymptotically optimal in the sense that
\begin{align}
\lim_{n_w\rightarrow \infty}\lim_{K\rightarrow \infty}\frac{\mathbb{E}\left[\ell(w_4(X_1^{K}|Y_1^{K}))	\right]}{K}\leq H(\mathbf{X}|\mathbf{Y}).
\end{align}
\end{theorem}
\begin{proof}
Let $C_p$ be the number of phrases of $X_{n_w+1}^K$ parsed by Algorithm \ref{alg:var}.
As depicted before, $l_i$ signifies the length of the $i$-th phrase as
	\begin{align}
	l_1&=\max\bigg\{ n\leq K-n_w~:~\exists ~t\in[1:n_w]\text{ s.t. }(XY)_{n_w+1}^{n_w+n}=(XY)_{n_w+1-t}^{n_w+n-t}	\bigg\},\\
	l_i&=\max\bigg\{ n\leq K-u_i+1~:~\exists~t\in[1:n_w]\text{ s.t. }(XY)_{u_i}^{u_i+n-1}=(XY)_{u_i-t}^{u_i+n-1-t}\bigg\},~\forall~i\in[2:C_p],
	\end{align}
where $u_i$ is defined as
	\begin{align}
	u_1&=n_w+1,\\
	u_i&=u_{i-1}+l_{i-1},~\forall~i\in[2:C_p].
	\end{align}
For some fixed value $\epsilon>0$, define $l_0$ as
	\begin{align}
	l_0=\left\lceil\frac{\log_2 n_w}{H(\mathbf{X},\mathbf{Y})+\epsilon}\right\rceil.	
	\end{align}
Divide the interval $[n_w+1:K]$ into subintervals of length 	$l_0$:
	\begin{align*}
	\mathcal{I}_1=[n_w+1:n_w+l_0],~\mathcal{I}_2=[n_w+l_0+1:n_w+2l_0],\cdots.	
	\end{align*}
The length of the last subinterval may be shorter than $l_0$.
The number of subintervals is $\lceil \frac{K-n_w}{l_0}\rceil$.
Define $\mathcal{F}$ as a set with indices of the phrases which are totally included by a subinterval and of which the next symbol is in the same subinterval, i.e., 
	\begin{align}
	\mathcal{F}=\bigg\{i\leq C_p~:~\exists ~t\in\left[1:	\left\lceil \frac{K-n_w}{l_0}\right\rceil\right] \text{ s.t. }[u_i:u_i+l_i]\subset \mathcal{I}_t\bigg\}.
	\end{align}
Then, the complement of the set is
	\begin{align}
	\mathcal{F}^c&=\left\{i:~	[u_i:u_i+l_i]\not\subset\mathcal{I}_t~\forall~t\right\}.
	\end{align}	
Since for all $i\in\mathcal{F}^c$
	\begin{align}
	\exists~t\text{ s.t. }\max 	\mathcal{I}_t\in [u_i:u_i+l_i-1],
	\end{align}
that is to say, the $i$-th phrase $[u_i:u_i+l_i-1]$ must include the last position of some subinterval,
the cardinality of $\mathcal{F}^c$ satisfies 
	\begin{align}
	|\mathcal{F}^c| \leq 	\left\lceil \frac{K-n_w}{l_0}\right\rceil\leq  \frac{K}{l_0}.	\label{eqn:comp}
	\end{align}

As the prefix code $g(\cdot)$ satisfies \eqref{eqn:g}, the average codelength can be bounded as	
	\begin{align}
	&\mathbb{E}\left[\ell\left(w_4\left(X_1^K|Y_1^K\right)\right)\right]\nonumber\\
	&=\lceil n_w\log_2|\mathcal{A}|\rceil + \mathbb{E}\left[\sum_{i=1}^{C_p}\min\left( \lceil \log_2 c_i\rceil,\lceil l_i\log_2|\mathcal{A}|\rceil\right) \right] +\mathbb{E}\left[\sum_{i=1}^{C_p} \ell\left(g(l_i)\right)\right] 	\\
	&\leq \lceil n_w\log_2|\mathcal{A}|\rceil+\gamma_1 \mathbb{E}\left[\sum_{i\in\mathcal{F}}l_i\right] +\gamma_2 \mathbb{E}\left[\sum_{i\in\mathcal{F}^c}\log_2 (l_i+1)\right]+ \mathbb{E}\left[\sum_{i\in\mathcal{F}^c}(\log_2 c_i+1)\right]\label{eqn:sum}
	\end{align}
for some constants $\gamma_1,~\gamma_2$.
The second term of \eqref{eqn:sum} is
	\begin{align}
	 &\mathbb{E}\left[\sum_{i\in\mathcal{F}}l_i\right] \nonumber\\
	 &\leq l_0\mathbb{E}\left[\left|\left\{t~:~[u_i:u_i+l_i-1]\subset \mathcal{I}_t~\exists~i\in\mathcal{F}\right\}\right|\right]	\label{eqn:size}\\
	 &\leq l_0 	 \left\lceil\frac{K-n_w}{l_0}\right\rceil\mathbb{P}\left[T_{l_0,1}(\mathbf{X},\mathbf{Y})>n_w\right]\label{eqn:p}\\
	&\leq K\mathbb{P}\left[T_{l_0,1}(\mathbf{X},\mathbf{Y})>n_w\right],\label{eqn:2nd}
	\end{align}
where \eqref{eqn:size} holds since $i\in\mathcal{F}$ implies $l_i<l_0$. \eqref{eqn:p} holds because the recurrence time of the subinterval that includes the $i$-th phrase is not within the window.
The third term of \eqref{eqn:sum} is bounded as
	\begin{align}
	&\mathbb{E}\left[\sum_{i\in\mathcal{F}^c}\log_2 (l_i+1)\right]\nonumber\\
	&\leq\mathbb{E}\left[ |\mathcal{F}^c| \log_2 \left(\frac{1}{|\mathcal{F}^c|}\sum_{i\in\mathcal{F}^c} l_i+1\right)\right]\label{eqn:je}\\
	&\leq \mathbb{E}\left[ |\mathcal{F}^c|\log_2 \left(\frac{K}{|\mathcal{F}^c|} +1\right)\right]\\
	&\leq \frac{K}{l_0}\log_2(l_0+1)\label{eqn:j2}
	\end{align}
where \eqref{eqn:je} follows by Jensen's inequality and \eqref{eqn:j2} follows by \eqref{eqn:comp}.
In order to find an upper bound of the fourth term, partition $\mathcal{F}^c$ into three sets:
	\begin{align}
	\mathcal{G}_1&=\left\{i\in\mathcal{F}^c~:~l_i<l_0\right\},\\
	\mathcal{G}_2&=\left\{i\in\mathcal{F}^c~:~l_i\geq l_0,~c_i>2^{l_0 (H(\mathbf{X}|\mathbf{Y})+\epsilon})\right\},\\
	\mathcal{G}_3&=\left\{i\in\mathcal{F}^c~:~l_i\geq l_0,~c_i\leq 2^{l_0(H(\mathbf{X}|\mathbf{Y})+\epsilon)}\right\}.	
	\end{align}
For any $t\in\left[1:\left\lceil \frac{K-n_w}{l_0}\right\rceil\right]$, define a function $f$ as
	\begin{align}
	f(t)=i\in[1:C_p]~\text{s.t.}~ \max \mathcal{I}_t\in [u_i:u_i+l_i-1].	
	\end{align}
$f(t)$ is well defined, because parsing is continuous and no phrase overlaps with another.
The average cardinality of $\mathcal{G}_1$ is bounded as
	\begin{align}
	\mathbb{E}[|\mathcal{G}_1|]&=\mathbb{E}\left[\sum_{i\in\mathcal{F}^c}\mathbf{1}\{i\in\mathcal{G}_1\}\right]\\
	&\leq \mathbb{E}\left[\sum_{t=1}^{\left\lceil \frac{K-n_w}{l_0}\right\rceil} \mathbf{1}\{f(t)\in\mathcal{G}_1\}\right]\label{eqn:stew}\\
	&=\sum_{t=1}^{\left\lceil \frac{K-n_w}{l_0}\right\rceil} \mathbb{P}[l_{f(t)}<l_0]\\
	&\leq \sum_t \mathbb{P}[T_{l_0,1}(\mathbf{X},\mathbf{Y})>n_w]\\
	&\leq \frac{K}{l_0}\mathbb{P}[T_{l_0,1}(\mathbf{X},\mathbf{Y})>n_w],
	\end{align}
where \eqref{eqn:stew} holds because for any $i\in\mathcal{F}^c$ there exists $t$ such that $f(t)=i$.
The number of $\mathbf{Y}$-matches within the window cannot exceed the window size $n_w$, and hence,
	\begin{align}
	&\mathbb{E}\left[ \sum_{i\in\mathcal{G}_1}\log_2 c_i\right]	\leq \log_2 n_w\frac{K}{l_0} \mathbb{P}\left[ T_{l_0,1}(\mathbf{X},\mathbf{Y})>n_w\right].\label{eqn:t1}
	\end{align}
Similarly, we have an upper bound of the summation over $\mathcal{G}_2$ as
	\begin{align}
	&\mathbb{E}\left[\sum_{i\in\mathcal{G}_2}\log_2 c_i\right] \leq \log_2 n_w \mathbb{E}\left[ |\mathcal{G}_2|\right]\\
	&\leq \log_2 n_w \frac{K}{l_0}\mathbb{P}\left[ T_{l_0,2^{l_0(H(\mathbf{X}|\mathbf{Y})+\epsilon)}+1}(\mathbf{Y})\leq n_w\right]\label{eqn:t2}.
	\end{align}
$i\in\mathcal{G}_3$ implies $\log_2c_i\leq l_0(H(\mathbf{X}|\mathbf{Y})+\epsilon)$, and
	\begin{align}
	&\mathbb{E}\left[ \sum_{i\in\mathcal{G}_3} \log_2 c_i\right]\leq l_0(H(\mathbf{X}|\mathbf{Y})+\epsilon)\mathbb{E}\left[|\mathcal{G}_3|\right]\\
	&\leq K(H(\mathbf{X}|\mathbf{Y})+\epsilon).\label{eqn:t3}
	\end{align}
	
Combining all the bounds, the asymptotic data compression rate is
 	\begin{align}
	&\lim_{n_w\rightarrow \infty}\lim_{K\rightarrow \infty}\frac{\mathbb{E}\left[\ell(w_4(X_1^{K}|Y_1^{K}))	\right]}{K}\nonumber\\
	&\leq 
	\gamma_1 \lim_{n_w\rightarrow \infty} \mathbb{P}\left[ T_{l_0,1}(\mathbf{X},\mathbf{Y})>n_w\right]+\gamma_2\lim_{n_w\rightarrow \infty} \frac{\log_2(l_0+1)}{l_0} +\lim_{n_w\rightarrow\infty} \frac{\log_2 n_w}{l_0}\mathbb{P}\left[ T_{l_0,1}(\mathbf{X},\mathbf{Y})>n_w\right]\nonumber\\
	&+\lim_{n_w\rightarrow\infty} \frac{\log_2 n_w}{l_0}\mathbb{P}\left[ T_{l_0,2^{l_0(H(\mathbf{X}|\mathbf{Y})+\epsilon)}+1}(\mathbf{Y})\leq n_w\right]+H(\mathbf{X}|\mathbf{Y})+\epsilon,
	\end{align}
	for any $\epsilon>0$.
Since
	\begin{align}
	\lim_{n_w\rightarrow\infty} \frac{1}{l_0}\log n_w >H(\mathbf{X},\mathbf{Y}),
	\end{align}
\cite[Corollary 3.4]{WZ1994} leads to
	\begin{align}
	\lim_{n_w\rightarrow\infty} \mathbb{P}\left[ T_{l_0,1}(\mathbf{X},\mathbf{Y})>n_w\right]=0.
	\end{align}
For any unbounded increasing sequence $J_n$, we have \cite{A1999}
	\begin{align}
	\lim_{n\rightarrow\infty}\frac{1}{n}\log_2\frac{T_{n,J_n}(\mathbf{Y})}{J_n}=H(\mathbf{Y})\quad\text{a.s.}	\label{eqn:algoet}
	\end{align}
Therefore, applying $J_n=2^{n(H(\mathbf{X}|\mathbf{Y})+\epsilon)}+1$ to \eqref{eqn:algoet},
	\begin{align}
	\lim_{n_w\rightarrow\infty} 	\mathbb{P}\left[ T_{l_0,2^{l_0(H(\mathbf{X}|\mathbf{Y})+\epsilon)}+1}(\mathbf{Y})\leq n_w\right]=0,
	\end{align}
which completes the proof.

\end{proof}

\begin{figure}[!t]
\centering
 \includegraphics[width=.6\columnwidth]{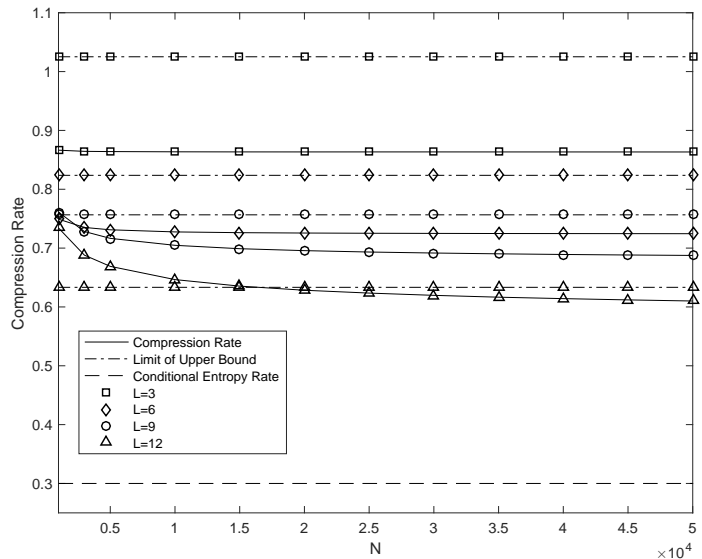}\\
 \caption{Compression rates of Algorithm \ref{alg:main} for varying $N$.}
 \label{fig:main}
\end{figure}

\begin{figure}[!t]
\centering
 \includegraphics[width=.6\columnwidth]{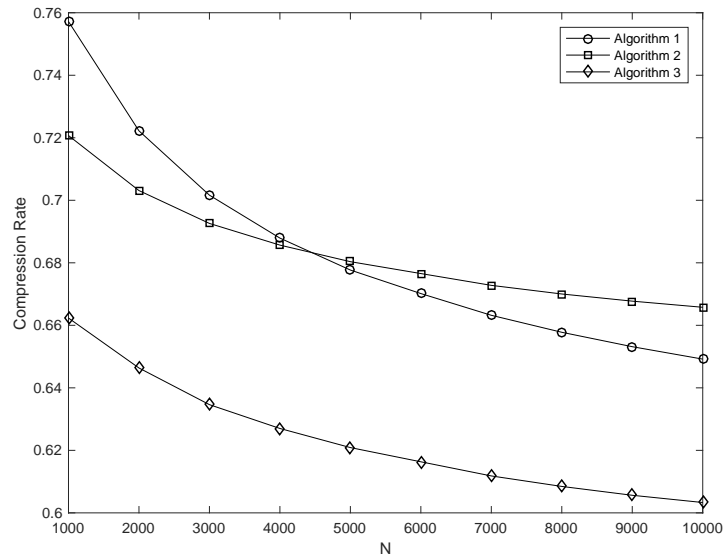}\\
 \caption{Compression rates of Algorithms \ref{alg:main}, \ref{alg:mod1}, \ref{alg:mod2} when $L=15$.}
 \label{fig:mod1}
\end{figure}

\begin{figure}[!t]
\centering
 \includegraphics[width=.6\columnwidth]{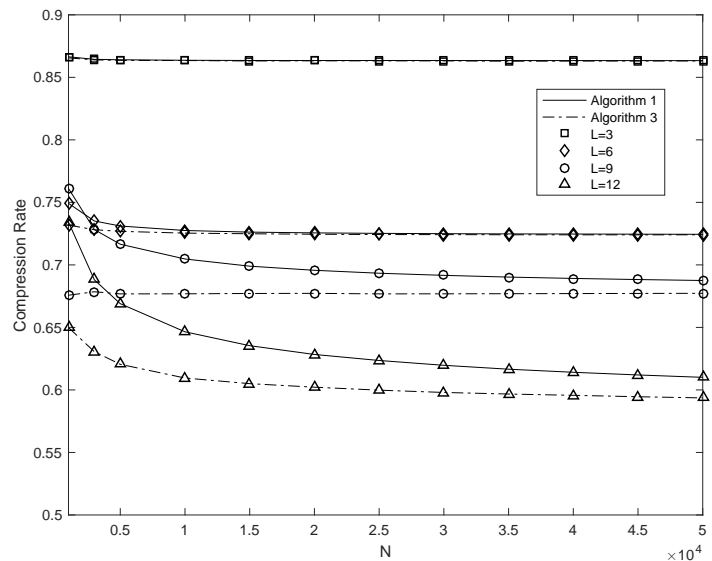}\\
  \caption{Compression rates of Algorithm \ref{alg:main} and \ref{alg:mod2} parametrized by L.}
  \label{fig:mod2}
\end{figure}

\section{Simulation} \label{section:sim}

In this section, simulation results are provided considering Markov chain model for $\mathcal{A}=\mathcal{B}=\{0,1\}$.
We have four states $(X_i,Y_i)=(0,0),~(0,1),~(1,0),~(1,1)$, and compare the proposed algorithms with the transition matrix
	\begin{align*}
	\left[\begin{array}{cccc} q &.25 &.25 &\frac{1-q}{3}\\ \frac{1-q}{3} &.25& .25&\frac{1-q}{3}\\ \frac{1-q}{3} &.25& .25&\frac{1-q}{3}\\\frac{1-q}{3} &.25& .25& q \end{array}\right].
	\end{align*}

Fig. \ref{fig:main} shows the result of compression rate when Algorithm \ref{alg:main} is used and $q=0.9$.
The rate versus $N$, the number of blocks, is shown with \eqref{eqn:limit}, which is the limit of upper bound of the compression rate.
Note that for every $L$, the rate converges to the values below upper bounds.
Also, as $L$ increases, the limit approaches the conditional entropy rate.

In Fig. \ref{fig:mod1}, compression rates of three algorithms are compared.
Algorithm \ref{alg:mod1} is more beneficial than Algorithm \ref{alg:main} is when source is relatively shorter, 
since beginning blocks satisfy the condition \eqref{eqn:suff}.
As $N$ grows, Algorithm \ref{alg:main} gradually outstrips, for blocks start to be assigned with longer codeword from Algorithm \ref{alg:mod1} than codeword from Algorithm \ref{alg:main}.
Algorithm \ref{alg:mod2} outperforms Algorithm \ref{alg:main} for any $N$ as proven before. 
In Fig. \ref{fig:mod2}, however, it is shown that since both algorithms yield the same codeword after some period, both curves become indistinguishable as $N$ goes to infinity.

\end{document}